\documentclass[journal]{IEEEtran}

\usepackage{amsmath,amssymb,amsfonts}
\usepackage{mathtools}
\usepackage{bm}
\usepackage{algorithm}
\usepackage{algorithmic}
\usepackage{booktabs}
\usepackage{graphicx}
\usepackage{url}
\usepackage{cite}
\usepackage{xcolor}
\newcommand{\rev}[1]{\textcolor{black}{#1}}

\usepackage{balance}

\newtheorem{proposition}{Proposition}
\newtheorem{theorem}{Theorem}
\newtheorem{remark}{Remark}

\newcommand{\safeincludegraphics}[2][]{%
  \IfFileExists{#2}{\includegraphics[#1]{#2}}{%
    \fbox{\parbox[c][1.55in][c]{0.94\linewidth}{\centering
    Missing figure file:\\[2pt]\texttt{\detokenize{#2}}}}}}

\begin{document}

\title{NOMA-Assisted Multi-User Hybrid Wireless-Fed Pinching-Antenna Systems}
\author{Hui Yang, Peng Zhu, Ming Zeng, \textit{Senior Member, IEEE}, Ebrahim Bedeer, Saeid Pakravan and Yulei Wang
\thanks{The work of M. Zeng was supported in part by NSERC under Grants RGPIN-2026- 05442 and CRC-2022-00115, and in part by FRQNT under Grant 341270. (Corresponding author: Ming Zeng.)} 

    \thanks{H. Yang is with Hunan Institute of Science and Technology, Yueyang, China (email: achelyal@163.com).}
    
    \thanks{P. Zhu is with Hunan Institute of Science and Technology, Yueyang, China (email: zhupeng@hnist.edu.cn).}
    
    \thanks{M. Zeng is with Laval University, Quebec City, Canada (email: ming.zeng@gel.ulaval.ca).}

    \thanks{E. Bedeer is with University of Saskatchewan, Saskatoon, SK, Canada (email: e.bedeer@usask.ca).}

    \thanks{S. Pakravan is with University of Quebec at Montreal, Montreal, QC, Canada (email: pakravan.saeid@uqam.ca).}

    \thanks{Y. Wang is with South-Central Minzu University, Wuhan, China (email: ylwang@mail.scuec.edu.cn).}

    %\thanks{N. Xia is with Nanjing Normal University, Nanjing, China (e-mail: nian.xia@nnu.edu.cn).} 
    }
\maketitle

\begin{abstract}
This paper investigates a non-orthogonal multiple-access (NOMA)-assisted multi-user wireless-fed pinching-antenna system (Wi-PASS). A multi-antenna base station (BS) simultaneously serves one direct user and wirelessly feeds a full-duplex amplify-and-forward relay equipped with a directional horn receiver. The relay injects the NOMA waveform into a dielectric waveguide, and one position-adjustable pinching antenna serves two additional users. Under maximum-gain zero-forcing transmission, ideal successive interference cancellation, and an additive residual self-interference model, we minimize the total consumed power by jointly optimizing the BS powers, relay amplification factor, NOMA power coefficients, decoding order, and pinching-antenna position. For a fixed position and decoding order, a variable transformation reduces the resource-allocation problem to a strictly convex scalar problem and yields a closed-form global solution. The position-dependent decoding order partitions the waveguide into finitely many intervals, and the derivative of the optimized power is governed by a quadratic polynomial on each interval. Hence, the globally optimal position is found by evaluating a small finite candidate set. Simulations at 28~GHz over 1000 random user topologies show that the proposed architecture consistently requires the lowest consumed power among direct-transmission, array-fed, no-PASS, and equal-time orthogonal multiple-access (OMA) benchmarks. At target signal-to-interference-plus-noise ratios of 20, 25, and 30~dB, it reduces the average power by 19.2\%, 21.1\%, and 21.7\%, respectively, relative to equal-time OMA, while preserving its advantage as the BS--relay distance and residual SI increase.
\end{abstract}


\section{Introduction}
\label{sec:introduction}

Millimeter-wave and sub-terahertz communication systems provide abundant spectrum resources and enable high-data-rate transmission. 
However, the severe propagation loss and high susceptibility to blockage at these frequencies substantially limit reliable coverage, particularly when the transmitter and intended users are separated by long or obstructed links \cite{Sun_TVT18, Hao_Network22}.
%Their coverage, however, is severely constrained by free-space attenuation and blockage \cite{Sun_TVT18, Hao_Network22}. 
Pinching-antenna systems (PASS) have recently emerged as a flexible waveguide-based architecture for mitigating these propagation limitations \cite{ding2024flexible}. In PASS, the signal is conveyed along a dielectric waveguide and radiated through one or more pinching antennas \rev{(PAs)} placed at favorable locations. By moving the radiation points closer to the intended users, PASS can shorten the final free-space propagation distance and improve link quality \cite{fukuda2022pinching, yang2025principles, Zeng_WCM26}. Unlike conventional fixed antenna arrays, PASS provides an additional degree of spatial reconfigurability by adjusting the physical locations or activation states of the PAs along the waveguide. Motivated by this spatial flexibility, existing studies have developed physical and signal models for PASS, characterized its performance, and investigated joint optimization of transmission resources and PA locations\cite{liu2025architecture,tyrovolas2026performance,wang2025modeling, Zhao_TCOM25, Zhao_pass1}. 
%PASS can reconfigure the wireless channel by adjusting the positions or activation states of the pinching antennas. Motivated by these advantages, recent studies have investigated PASS architectures, physical and signal models, performance characterization, and joint transmit-resource and antenna-position optimization \cite{liu2025architecture,tyrovolas2026performance,wang2025modeling, Zhao_TCOM25, Zhao_pass1}.

Beyond coverage enhancement, the shared-feed structure of PASS provides a natural platform for multi-user transmission. Since multiple signals can be conveyed through the same waveguide and radiated from shared PAs, power-domain non-orthogonal multiple access (NOMA) can be employed to multiplex multiple users over the same time-frequency resource. The users' signals are then separated through successive interference cancellation (SIC) \cite{ding2024flexible, wang2024, Fu_GC25}. Accordingly, both NOMA- and orthogonal multiple-access (OMA)-based PASS designs have been studied to improve spectral efficiency, energy efficiency, and user connectivity \cite{cheng2025omanoma, Zeng_COMML25, zeng2025EE, KE2026103249}. However, these studies generally assume that the dielectric waveguide is physically connected to the base station (BS).

Although wired feeding can provide a reliable signal source for PASS, it may become costly or impractical when the waveguide must be deployed far from the BS, around obstacles, or in distributed coverage regions. Wireless-fed PASS (Wi-PASS) addresses this deployment limitation by delivering the signal to a remote waveguide through a wireless first hop \cite{wijewardhana2025wipass}. Wi-PASS therefore combines the deployment flexibility of wireless relaying with the favorable final-hop propagation offered by movable pinching antennas. A horn-assisted full-duplex amplify-and-forward Wi-PASS architecture was subsequently investigated for single-user transmission, demonstrating that directional horn gain and passive isolation can improve the wireless feeding link \cite{feng2026horn}. Full-duplex operation, nevertheless, introduces residual self-interference (SI), while signal amplification also increases thermal noise and relay power consumption \cite{sabharwal2014full}.

\rev{Extending Wi-PASS from single-user transmission to NOMA-based multi-user service introduces several design challenges. The power available to the PASS-side users depend jointly on the wireless feeding link, the relay amplification factor, and the subsequent waveguide-assisted transmission. Introducing NOMA further couples the users' power allocation, SIC decoding order, and PA position. In addition, the relay amplification factor scales both the desired first-hop signal and the relay thermal-noise and residual-SI components, thereby coupling the BS relay-feeding power with the PASS-side signal-to-interference-plus-noise ratios (SINRs). Consequently, the BS relay-feeding power, relay amplification factor, NOMA power coefficients, SIC decoding order, and PA position should be jointly optimized.}

\rev{The original Wi-PASS framework in \cite{wijewardhana2025wipass} introduced wireless feeding of a remote PASS, while the horn-assisted full-duplex design in \cite{feng2026horn} considered residual-SI-aware optimization for single-user transmission. These studies do not address the hybrid multi-user setting considered here, where a zero-forcing (ZF)-enabled BS simultaneously serves one user directly and feeds a full-duplex relay, and two additional users share a common position-adjustable PA through power-domain NOMA. The resulting optimization problem has coupled decision variables comprising the BS transmit powers, relay amplification factor, NOMA power allocation, SIC decoding order, and PA position. To distinguish the sources of power reduction, the numerical evaluation separately compares the effects of the horn-assisted first hop, PASS-enabled second hop, and NOMA transmission.}

The main contributions of this work are summarized as follows:
\begin{itemize}
\item \rev{\textbf{Hybrid Wi-PASS architecture and power model:} We formulate a hybrid multi-user Wi-PASS architecture in which a ZF-enabled BS simultaneously serves one direct user and wirelessly feeds a full-duplex relay, while two PASS-side users share one position-adjustable PA through power-domain NOMA. The total consumed-power model incorporates the BS and relay amplifier efficiencies, circuit power, amplified thermal noise, and additive residual SI.}

\item \textbf{Closed-form resource allocation:} \rev{For a fixed PA position and SIC decoding order, we introduce amplification-weighted power variables that transform the quality-of-service (QoS) constraints into affine inequalities. The resulting resource-allocation problem reduces to a strictly convex scalar problem whose unique optimum is available in closed form, yielding the optimal NOMA powers and relay amplification factor.}

\item \textbf{Finite-candidate global PA positioning:}  \rev{We characterize the dependence of the optimal SIC decoding order on the PA position and show that the decoding order can change at most at two position-switching points. Within each interval with a fixed decoding order, every stationary PA position is a root of a quadratic polynomial. The globally optimal position can therefore be found by evaluating a finite candidate set without requiring dense sampling over the feasible PA-position range.}

\item \textbf{Benchmark-driven numerical insights:} \rev{We develop benchmark schemes comprising direct transmission, array-fed Wi-PASS, horn-assisted relay without PASS, array-assisted relay without PASS, and equal-time OMA-based Wi-PASS. In the OMA benchmark, the two PASS-side users are assigned equal time fractions and are served through a jointly optimized common PA position.}
Monte Carlo results quantify the distinct benefits of the proposed architecture. In the SINR sweep, the proposed NOMA design reduces the total consumed power by up to 21.7\% relative to equal-time OMA and by approximately 40\% relative to array-fed Wi-PASS. The distance and residual-SI sweeps further demonstrate the benefits provided by the directional horn-assisted wireless first hop and the PASS-enabled second hop under different operating conditions.

\end{itemize}

\rev{\textit{Notation:} Bold lowercase and uppercase symbols denote vectors and matrices, respectively, while plain symbols denote scalars. The superscript $(\cdot)^H$ denotes the Hermitian transpose, $\|\cdot\|$ is the Euclidean norm, $\mathbb{E}[\cdot]$ is expectation, $\mathcal{CN}(\mu,\sigma^2)$ denotes a circularly symmetric complex Gaussian distribution with mean $\mu$ and variance $\sigma^2$, and $\bm{I}$ is an identity matrix of appropriate dimension. The operator $(\cdot)^\dagger$ denotes the Moore--Penrose pseudoinverse, and $\bm{\Pi}^{\perp}$ denotes an orthogonal projection matrix onto the orthogonal complement of the subspace specified by its subscript.}

The remainder of this paper is organized as follows. Section~\ref{sec:system_model} presents the system model and problem formulation. Section~\ref{sec:solution} derives the closed-form resource-allocation solution and the finite-candidate method for globally optimizing the PA position. Section~\ref{sec:benchmarks} introduces the benchmark schemes. Section~\ref{sec:simulation} presents the numerical results, and Section~\ref{sec:conclusion} concludes the paper.

\section{System Model and Problem Formulation}
\label{sec:system_model}

\begin{figure}[t]
\centering
\safeincludegraphics[width=0.5\textwidth]{Images/system.png}
\caption{System model of the proposed NOMA-assisted hybrid Wi-PASS, where the BS serves a direct user \(U_d\) and wirelessly feeds a full-duplex AF relay. }
\label{fig:system model}
\end{figure}

\subsection{Network Architecture and BS Transmission}

As shown in Fig. \ref{fig:system model}, we consider a multi-user hybrid Wi-PASS system, in which an $N_{\mathrm{t}}$-antenna BS simultaneously serves one direct user, denoted by $U_{\mathrm{d}}$, and wirelessly feeds a full-duplex amplify-and-forward relay. The relay employs a directional horn receiving antenna and injects the amplified signal into a dielectric waveguide. \rev{A position-adjustable PA radiates the guided signal toward two PASS-side users, denoted by $U_1$ and $U_2$, which are multiplexed through downlink NOMA.}

The PA is located at
\begin{equation}
    \bm{\Phi}_{\mathrm{p}}=(x,0,d), \qquad 0\leq x\leq L,
    \label{eq:pa_location}
\end{equation}
where $x$ is the PA position, $L$ is the waveguide length, and $d$ is the waveguide height. The position of PASS-side user $U_k$, $k\in\{1,2\}$, is
\begin{equation}
    \bm{\Phi}_k=(x_k,y_k,0).
    \label{eq:user_location}
\end{equation}

The BS simultaneously transmits a relay-feeding stream and a direct-user stream as
\begin{equation}
    \bm{x}_{\mathrm{B}}
    =\sqrt{P_{\mathrm{r}}}\bm{v}_{\mathrm{r}}s_{\mathrm{r}}
    +\sqrt{P_{\mathrm{d}}}\bm{v}_{\mathrm{d}}s_{\mathrm{d}},
    \label{eq:bs_signal}
\end{equation}
where $P_{\mathrm{r}}$ and $P_{\mathrm{d}}$ are the BS radiated powers allocated to relay feeding and direct transmission, respectively. The beamformers $\bm{v}_{\mathrm{r}}$ and $ \bm{v}_{\mathrm{d}}$ satisfy
\begin{equation}
    \|\bm{v}_{\mathrm{r}}\|^2=\|\bm{v}_{\mathrm{d}}\|^2=1,
    \label{eq:beamformer_norms}
\end{equation}
and $\mathbb{E}[|s_{\mathrm{d}}|^2]=1$.

For a given PA position, let $w\in\{1,2\}$ and $s\in\{1,2\}\setminus\{w\}$ denote the weak and strong NOMA users, respectively. The relay-feeding symbol is
\begin{equation}
    s_{\mathrm{r}}=\sqrt{a_w}s_w+\sqrt{a_s}s_s,
    \label{eq:noma_symbol}
\end{equation}
where $\mathbb{E}[|s_k|^2]=1$ and
\begin{equation}
    a_w+a_s=1, \qquad a_w\geq 0,\quad a_s\geq 0.
    \label{eq:noma_coefficients}
\end{equation}
The position-dependent decoding order is specified in Section~\ref{subsec:noma_order}.

{\color{black}
Let $\bm{h}_{\mathrm{BR}}\in\mathbb{C}^{N_{\mathrm{t}}\times 1}$ and $\bm{h}_{\mathrm{d}}\in\mathbb{C}^{N_{\mathrm{t}}\times 1}$ denote the BS--relay and BS--direct-user channels, respectively. We adopt ZF transmission such that
\begin{equation}
    \bm{h}_{\mathrm{BR}}^H\bm{v}_{\mathrm{d}}=0,
    \qquad
    \bm{h}_{\mathrm{d}}^H\bm{v}_{\mathrm{r}}=0.
    \label{eq:zf_conditions}
\end{equation}
For $N_{\mathrm{t}}\geq 2$ and linearly independent $\bm{h}_{\mathrm{BR}}$ and $\bm{h}_{\mathrm{d}}$, the maximum-gain ZF relay-feeding beamformer is
\begin{align}
    \bm{v}_{\mathrm{r}}
    &=\frac{\bm{\Pi}_{\mathrm{d}}^{\perp}\bm{h}_{\mathrm{BR}}}
    {\|\bm{\Pi}_{\mathrm{d}}^{\perp}\bm{h}_{\mathrm{BR}}\|},
    &
    \bm{\Pi}_{\mathrm{d}}^{\perp}
    &=\bm{I}-\frac{\bm{h}_{\mathrm{d}}\bm{h}_{\mathrm{d}}^H}
    {\|\bm{h}_{\mathrm{d}}\|^2},
    \label{eq:vr_zf}
\end{align}
with effective BS--relay gain
\begin{equation}
    G\triangleq |\bm{h}_{\mathrm{BR}}^H\bm{v}_{\mathrm{r}}|^2.
    \label{eq:G_definition}
\end{equation}
Similarly, the maximum-gain ZF direct-user beamformer is
\begin{align}
    \bm{v}_{\mathrm{d}}
    &=\frac{\bm{\Pi}_{\mathrm{BR}}^{\perp}\bm{h}_{\mathrm{d}}}
    {\|\bm{\Pi}_{\mathrm{BR}}^{\perp}\bm{h}_{\mathrm{d}}\|},
    &
    \bm{\Pi}_{\mathrm{BR}}^{\perp}
    &=\bm{I}-\frac{\bm{h}_{\mathrm{BR}}\bm{h}_{\mathrm{BR}}^H}
    {\|\bm{h}_{\mathrm{BR}}\|^2},
    \label{eq:vd_zf}
\end{align}
with effective BS--direct-user gain
\begin{equation}
    D\triangleq |\bm{h}_{\mathrm{d}}^H\bm{v}_{\mathrm{d}}|^2.
    \label{eq:D_definition}
\end{equation}
}

\subsection{Full-duplex Relay, Power Consumption, and PASS-Side Channels}

\rev{Residual SI after practical passive and active cancellation
arises from several hardware and channel impairments \cite{sabharwal2014full}. For analytical tractability, we model the aggregate residual SI as an additive zero-mean circularly symmetric complex Gaussian term, $z_{\mathrm{SI}}\sim\mathcal{CN}(0,\sigma_{\mathrm{SI}}^2)$, which is assumed to be independant of the relay transmit power.} The relay receives
\begin{equation}
    y_{\mathrm{R}}
    =\sqrt{P_{\mathrm{r}}}\bm{h}_{\mathrm{BR}}^H\bm{v}_{\mathrm{r}}s_{\mathrm{r}}
    +n_{\mathrm{R}}+z_{\mathrm{SI}},
    \label{eq:relay_received}
\end{equation}
where $n_{\mathrm{R}}\sim\mathcal{CN}(0,\sigma_{\mathrm{R}}^2)$ denotes the additive zero-mean circularly symmetric complex Gaussian noise at the relay. Define
\begin{equation}
    \widetilde{\sigma}_{\mathrm{R}}^2
    \triangleq \sigma_{\mathrm{R}}^2+\sigma_{\mathrm{SI}}^2.
    \label{eq:effective_relay_noise}
\end{equation}
The relay applies an amplification factor $\beta\geq 0$ and feeds
\begin{equation}
    t_{\mathrm{R}}=\beta y_{\mathrm{R}}
    \label{eq:relay_output_signal}
\end{equation}
into the dielectric waveguide, where $t_{\mathrm R}$ denotes the relay output waveform. Because both the relay thermal noise and the equivalent residual-SI component in \eqref{eq:relay_received} are amplified, the relay output power is
\begin{equation}
    P_{\mathrm{R}}^{\mathrm{out}}
    =\beta^2\left(P_{\mathrm{r}}G+\widetilde{\sigma}_{\mathrm{R}}^2\right).
    \label{eq:relay_output_power}
\end{equation}

Let $\eta_{\mathrm{B}}\in(0,1]$ and $\eta_{\mathrm{R}}\in(0,1]$ denote the BS and relay power-amplifier efficiencies, respectively. Let $P_{\mathrm{c,B}}$ and $P_{\mathrm{c,R}}$ denote their fixed circuit powers. The total consumed power is modeled as
\begin{equation}
    P_{\mathrm{tot}}
    =\frac{P_{\mathrm{r}}+P_{\mathrm{d}}}{\eta_{\mathrm{B}}}
    +\frac{P_{\mathrm{R}}^{\mathrm{out}}}{\eta_{\mathrm{R}}}
    +P_{\mathrm{c,B}}+P_{\mathrm{c,R}}.
    \label{eq:total_consumed_power}
\end{equation}
{\color{black}Accordingly, $P_{\mathrm{tot}}$ represents the power drawn by the BS and relay power amplifiers, together with their fixed circuit power consumption. Specifically, $P_r+P_d$ and $P_{\mathrm{R}}^{\mathrm{out}}$ denote the corresponding RF output powers, whereas division by $\eta_{\mathrm{B}}$ and $\eta_{\mathrm{R}}$ accounts for the power-amplifier inefficiencies. If amplifier inefficiencies and circuit consumption are neglected, i.e., $\eta_{\mathrm{B}}=\eta_{\mathrm{R}}=1$ and $P_{\mathrm{c,B}}=P_{\mathrm{c,R}}=0$, \eqref{eq:total_consumed_power} reduces to the corresponding transmit/output-power objective.}

Let $H_k(x)$ denote the power gain from the PA to PASS-side user $U_k$, given by
\begin{equation}
    H_k(x)
    =\frac{C_{\mathrm{P}}e^{-\alpha_{\mathrm{D}}x}}{\ell_k(x)},
    \qquad
    C_{\mathrm{P}}\triangleq
    \xi_{\mathrm{P}}\frac{c^2}{16\pi^2f^2},
    \label{eq:pass_channel_gain}
\end{equation}
where $\xi_{\mathrm{P}}\in(0,1]$ collects the waveguide-injection, pinching, and radiation efficiencies, and
\begin{equation}
    \ell_k(x)\triangleq (x-x_k)^2+y_k^2+d^2
    \label{eq:ellk_definition}
\end{equation}
is the squared PA--user distance. Here, $f$ is the carrier frequency, $c$ is the speed of light, and $\alpha_{\mathrm{D}}$ is the power attenuation coefficient. An attenuation value $\alpha_{\mathrm{dB}}$ specified in dB/m is converted according to $\alpha_{\mathrm{D}}=(\ln 10/10)\alpha_{\mathrm{dB}}$.

The direct user receives
\begin{equation}
    y_{\mathrm{d}}
    =\sqrt{P_{\mathrm{d}}}\bm{h}_{\mathrm{d}}^H\bm{v}_{\mathrm{d}}s_{\mathrm{d}}
    +n_{\mathrm{d}},
    \label{eq:direct_received}
\end{equation}
where $n_{\mathrm{d}}\sim\mathcal{CN}(0,\sigma_{\mathrm{d}}^2)$. Its signal-to-noise ratio is
\begin{equation}
    \gamma_{\mathrm{d}}=\frac{P_{\mathrm{d}}D}{\sigma_{\mathrm{d}}^2}.
    \label{eq:direct_sinr}
\end{equation}

\subsection{NOMA Decoding Order and Achievable SINRs}
\label{subsec:noma_order}

\rev{For a given $x$, define the PASS-side channel-gain-to-noise-power ratio}
\begin{equation}
    \zeta_k(x)\triangleq\frac{H_k(x)}{\sigma_k^2}.
    \label{eq:normalized_channel_metric}
\end{equation}
\rev{The user with the larger $\zeta_k(x)$ is designated as the strong user and performs SIC, while the other user is designated as the weak user. Thus, over an interval with a fixed decoding order,}
\begin{equation}
    \zeta_s(x)\geq \zeta_w(x),
    \label{eq:noma_order_metric}
\end{equation}
which is equivalent to
\begin{equation}
    \frac{\sigma_s^2}{H_s(x)}
    \leq
    \frac{\sigma_w^2}{H_w(x)}.
    \label{eq:noma_order_equivalent}
\end{equation}

The weak user decodes $s_w$ while treating $s_s$ as interference. Its SINR is
\begin{equation}
    \gamma_w
    =\frac{P_{\mathrm{r}}\beta^2Ga_wH_w(x)}
    {P_{\mathrm{r}}\beta^2Ga_sH_w(x)
    +\beta^2H_w(x)\widetilde{\sigma}_{\mathrm{R}}^2
    +\sigma_w^2}.
    \label{eq:weak_user_sinr}
\end{equation}
The SINR at the strong user for decoding the weak user's signal is
\begin{equation}
    \gamma_{s\rightarrow w}
    =\frac{P_{\mathrm{r}}\beta^2Ga_wH_s(x)}
    {P_{\mathrm{r}}\beta^2Ga_sH_s(x)
    +\beta^2H_s(x)\widetilde{\sigma}_{\mathrm{R}}^2
    +\sigma_s^2}.
    \label{eq:sic_sinr}
\end{equation}
After ideal SIC, the strong user decodes its own signal with SINR
\begin{equation}
    \gamma_s
    =\frac{P_{\mathrm{r}}\beta^2Ga_sH_s(x)}
    {\beta^2H_s(x)\widetilde{\sigma}_{\mathrm{R}}^2+\sigma_s^2}.
    \label{eq:strong_user_sinr}
\end{equation}

Let $\Gamma_k=2^{R_k^{\min}}-1$, $k\in\{1,2,\mathrm{d}\}$, denote the SINR targets associated with the minimum spectral efficiencies. Under the adopted maximum-gain ZF beamformers, additive residual-SI model, and ideal SIC, the total consumed-power minimization problem is
\begin{subequations}\label{prob:original}
\begin{align}
    \min_{\substack{P_{\mathrm{r}},P_{\mathrm{d}},\beta,\\a_w,a_s,x}}
    \quad &
    \frac{P_{\mathrm{r}}+P_{\mathrm{d}}}{\eta_{\mathrm{B}}}
    +\frac{\beta^2(P_{\mathrm{r}}G+\widetilde{\sigma}_{\mathrm{R}}^2)}
    {\eta_{\mathrm{R}}}
    +P_{\mathrm{c,B}}+P_{\mathrm{c,R}}
    \label{prob:original_obj}
    \\
    \mathrm{s.t.}\quad
    &\gamma_w\geq \Gamma_w,
    \label{prob:weak_qos}
    \\
    &\gamma_{s\rightarrow w}\geq \Gamma_w,
    \label{prob:sic_qos}
    \\
    &\gamma_s\geq \Gamma_s,
    \label{prob:strong_qos}
    \\
    &\gamma_{\mathrm{d}}\geq \Gamma_{\mathrm{d}},
    \label{prob:direct_qos}
    \\
    &a_w+a_s=1,\quad a_w,a_s\geq 0,
    \label{prob:noma_constraints}
    \\
    &P_{\mathrm{r}},P_{\mathrm{d}}\geq 0,\quad \beta\geq 0,
    \label{prob:power_nonnegative}
    \\
    &0\leq x\leq L,
    \label{prob:position_constraint}
\end{align}
\end{subequations}
where \rev{constraints \eqref{prob:weak_qos} and \eqref{prob:sic_qos} require the weak-user stream to be decodable at both $U_w$ and $U_s$, \eqref{prob:strong_qos} protects the strong user's own stream after SIC, and \eqref{prob:direct_qos} enforces the direct-user QoS. Problem~\eqref{prob:original} is nonconvex because these QoS constraints contain products of decision variables, such as $P_{\mathrm{r}}\beta^2a_k$, and also depend nonlinearly on the PA position through $H_k(x)$. We next remove the multiplicative coupling for a fixed position/order and then optimize the PA position globally.}
%\rev{Problem~\eqref{prob:original} is nonconvex because the QoS constraints contain products of decision variables, such as $P_{\mathrm{r}}\beta^2a_k$, and also depend nonlinearly on the PA position through $H_k(x)$. We next derive a closed-form inner solution and a finite-candidate global solution for the PA position.}

\section{Proposed Solution}
\label{sec:solution}

\subsection{ZF Beamforming and Direct-User Power}

Under \eqref{eq:zf_conditions}, the direct-user constraint is decoupled from the relay-feeding transmission. Since the objective is increasing in $P_{\mathrm{d}}$, its QoS constraint is active at optimum, yielding
\begin{equation}
    P_{\mathrm{d}}^{\star}
    =\frac{\Gamma_{\mathrm{d}}\sigma_{\mathrm{d}}^2}{D}.
    \label{eq:optimal_direct_power}
\end{equation}
\rev{The remaining decision variables are the relay-feeding power, relay amplification factor, NOMA coefficients, decoding order, and PA position.}

\subsection{Closed-Form Resource Allocation for Fixed Position and \rev{Decoding Order}}

\rev{For fixed $x$ and a fixed NOMA decoding order $(w,s)$, define}
\begin{equation}
    u\triangleq \beta^2,
    \qquad
    q_w\triangleq P_{\mathrm{r}}ua_w,
    \qquad
    q_s\triangleq P_{\mathrm{r}}ua_s.
    \label{eq:variable_transformation}
\end{equation}
Since $a_w+a_s=1$,
\begin{equation}
    q_w+q_s=P_{\mathrm{r}}u,
    \qquad
    P_{\mathrm{r}}=\frac{q_w+q_s}{u}.
    \label{eq:Pr_recovery}
\end{equation}
Define
\begin{equation}
    A\triangleq\frac{\widetilde{\sigma}_{\mathrm{R}}^2}{G},
    \qquad
    B_k(x)\triangleq\frac{\sigma_k^2}{GH_k(x)}.
    \label{eq:A_B_definitions}
\end{equation}
The three NOMA QoS constraints become
\begin{align}
    q_w&\geq \Gamma_w\big(q_s+Au+B_w(x)\big),
    \label{eq:weak_linear_constraint}
    \\
    q_w&\geq \Gamma_w\big(q_s+Au+B_s(x)\big),
    \label{eq:sic_linear_constraint}
    \\
    q_s&\geq \Gamma_s\big(Au+B_s(x)\big).
    \label{eq:strong_linear_constraint}
\end{align}
By \eqref{eq:noma_order_equivalent}, $B_w(x)\geq B_s(x)$; hence, \eqref{eq:weak_linear_constraint} implies \eqref{eq:sic_linear_constraint}.

\begin{proposition}\label{prop:fixed_x_solution}
Assume $G>0$, $D>0$, $H_k(x)>0$, and at least one PASS-side SINR target is positive. For fixed $x$ and a fixed NOMA \rev{decoding order} satisfying $B_w(x)\geq B_s(x)$, define
\begin{align}
    \Theta&\triangleq \Gamma_w+(1+\Gamma_w)\Gamma_s,
    \label{eq:Theta_definition}
    \\
    A_{\mathrm{t}}&\triangleq A\Theta,
    \label{eq:At_definition}
    \\
    B_{\mathrm{t}}(x)
    &\triangleq
    \Gamma_wB_w(x)
    +(1+\Gamma_w)\Gamma_sB_s(x).
    \label{eq:Bt_definition}
\end{align}
Then the globally optimal resource allocation for the fixed $(x,w,s)$ is
\begin{equation}
    u^{\star}(x)
    =\sqrt{\frac{\eta_{\mathrm{R}}B_{\mathrm{t}}(x)}
    {\eta_{\mathrm{B}}\left(GA_{\mathrm{t}}+\widetilde{\sigma}_{\mathrm{R}}^2\right)}},
    \label{eq:optimal_u}
\end{equation}
\begin{equation}
    q_s^{\star}(x)
    =\Gamma_s\big(Au^{\star}(x)+B_s(x)\big),
    \label{eq:optimal_qs}
\end{equation}
\begin{equation}
    q_w^{\star}(x)
    =\Gamma_w\big(q_s^{\star}(x)+Au^{\star}(x)+B_w(x)\big),
    \label{eq:optimal_qw}
\end{equation}
\begin{equation}
    P_{\mathrm{r}}^{\star}(x)
    =\frac{q_w^{\star}(x)+q_s^{\star}(x)}{u^{\star}(x)},
    \label{eq:optimal_Pr}
\end{equation}
\begin{equation}
    \beta^{\star}(x)=\sqrt{u^{\star}(x)},
    \label{eq:optimal_beta}
\end{equation}
and
\begin{equation}
    a_k^{\star}(x)
    =\frac{q_k^{\star}(x)}{q_w^{\star}(x)+q_s^{\star}(x)},
    \qquad k\in\{w,s\}.
    \label{eq:optimal_noma_coefficients}
\end{equation}
The corresponding minimum total consumed power is
\begin{align}
    P_{\mathrm{tot}}^{\star}(x;w,s)
    =&\;\frac{P_{\mathrm{d}}^{\star}}{\eta_{\mathrm{B}}}
    +P_{\mathrm{c,B}}+P_{\mathrm{c,R}}
    +\frac{A_{\mathrm{t}}}{\eta_{\mathrm{B}}}
    +\frac{GB_{\mathrm{t}}(x)}{\eta_{\mathrm{R}}}
    \nonumber\\
    &+2\sqrt{\frac{B_{\mathrm{t}}(x)
    \left(GA_{\mathrm{t}}+\widetilde{\sigma}_{\mathrm{R}}^2\right)}
    {\eta_{\mathrm{B}}\eta_{\mathrm{R}}}}.
    \label{eq:optimized_total_power}
\end{align}
\end{proposition}

\begin{IEEEproof}
For any fixed $u>0$, the objective is increasing in $q_w$ and $q_s$. Consequently, the minimum feasible values satisfy \eqref{eq:strong_linear_constraint} and \eqref{eq:weak_linear_constraint} with equality, which gives
\begin{equation}
    q_s^{\star}(u)=\Gamma_s(Au+B_s),
    \label{eq:qs_of_u}
\end{equation}
and
\begin{equation}
    q_w^{\star}(u)=\Gamma_w(q_s^{\star}(u)+Au+B_w).
    \label{eq:qw_of_u}
\end{equation}
Therefore,
\begin{equation}
    q_w^{\star}(u)+q_s^{\star}(u)
    =A_{\mathrm{t}}u+B_{\mathrm{t}}.
    \label{eq:qsum_affine}
\end{equation}
Using \eqref{eq:Pr_recovery} and \eqref{eq:relay_output_power}, the variable part of the total consumed power becomes
\begin{align}
    \phi(u)
    =&\frac{A_{\mathrm{t}}u+B_{\mathrm{t}}}{\eta_{\mathrm{B}}u}
    +\frac{G(A_{\mathrm{t}}u+B_{\mathrm{t}})
    +u\widetilde{\sigma}_{\mathrm{R}}^2}{\eta_{\mathrm{R}}}
    \nonumber\\
    =&\;\frac{A_{\mathrm{t}}}{\eta_{\mathrm{B}}}
    +\frac{GB_{\mathrm{t}}}{\eta_{\mathrm{R}}}
    +\frac{B_{\mathrm{t}}}{\eta_{\mathrm{B}}u}
    +\frac{GA_{\mathrm{t}}+\widetilde{\sigma}_{\mathrm{R}}^2}
    {\eta_{\mathrm{R}}}u.
    \label{eq:phi_u}
\end{align}
Since
\begin{equation}
    \phi''(u)=\frac{2B_{\mathrm{t}}}{\eta_{\mathrm{B}}u^3}>0,
    \label{eq:phi_second_derivative}
\end{equation}
$\phi(u)$ is strictly convex over $u>0$. Setting $\phi'(u)=0$ yields \eqref{eq:optimal_u}. Substitution into \eqref{eq:qs_of_u}--\eqref{eq:qw_of_u} gives \eqref{eq:optimal_qs}--\eqref{eq:optimal_noma_coefficients}, while substituting $u^{\star}$ into \eqref{eq:phi_u} yields \eqref{eq:optimized_total_power}.
\end{IEEEproof}

\subsection{Piecewise Closed-Form PA-Position Optimization}

\rev{The NOMA decoding order changes only when} the two effective metrics in \eqref{eq:normalized_channel_metric} are equal. From \eqref{eq:pass_channel_gain}, the common factors $C_{\mathrm{P}}e^{-\alpha_{\mathrm{D}}x}$ cancel, and the \rev{decoding-order switching equation} is
\begin{equation}
    \frac{H_1(x)}{\sigma_1^2}
    =\frac{H_2(x)}{\sigma_2^2}
    \quad\Longleftrightarrow\quad
    \sigma_1^2\ell_1(x)=\sigma_2^2\ell_2(x).
    \label{eq:order_switch_equation}
\end{equation}
Expanding \eqref{eq:order_switch_equation} gives
\begin{equation}
    a_{\mathrm{o}}x^2+b_{\mathrm{o}}x+c_{\mathrm{o}}=0,
    \label{eq:order_switch_quadratic}
\end{equation}
where
\begin{align}
    a_{\mathrm{o}}&=\sigma_1^2-\sigma_2^2,
    \label{eq:ao}
    \\
    b_{\mathrm{o}}&=-2(\sigma_1^2x_1-\sigma_2^2x_2),
    \label{eq:bo}
    \\
    c_{\mathrm{o}}&=\sigma_1^2(x_1^2+y_1^2+d^2)
    -\sigma_2^2(x_2^2+y_2^2+d^2).
    \label{eq:co}
\end{align}
Only real roots in $(0,L)$ are retained. When $\sigma_1^2=\sigma_2^2$ and $x_1\neq x_2$, \eqref{eq:order_switch_quadratic} reduces to the single switching point
\begin{equation}
    x_{\mathrm{sw}}
    =\frac{x_2^2+y_2^2-x_1^2-y_1^2}{2(x_2-x_1)}.
    \label{eq:equal_noise_switching_point}
\end{equation}
When $\sigma_1^2=\sigma_2^2$ and $x_1=x_2$, the \rev{decoding order never changes} if $y_1^2\neq y_2^2$, whereas \rev{both decoding orders are equivalent} for every $x$ if $y_1^2=y_2^2$.
\rev{Let $\mathcal{X}_{\mathrm{o}}$ denote the set of valid decoding-order switching points.} Sorting
\begin{equation}
    \{0,L\}\cup\mathcal{X}_{\mathrm{o}}
    =\{\tau_0,\tau_1,\ldots,\tau_M\},
    \qquad \tau_0<\tau_1<\cdots<\tau_M,
    \label{eq:interval_boundaries}
\end{equation}
partitions $[0,L]$ into intervals over which the \rev{decoding order $(w,s)$ is fixed}.
\rev{For a given fixed-decoding-order interval, define}
\begin{equation}
    \kappa_k\triangleq\frac{\sigma_k^2}{GC_{\mathrm{P}}}.
    \label{eq:kappa_definition}
\end{equation}
Then
\begin{equation}
    B_k(x)=\kappa_ke^{\alpha_{\mathrm{D}}x}\ell_k(x).
    \label{eq:Bk_explicit}
\end{equation}
Using \eqref{eq:Bt_definition}, define the positive weights
\begin{equation}
    \omega_w\triangleq\Gamma_w\kappa_w,
    \qquad
    \omega_s\triangleq(1+\Gamma_w)\Gamma_s\kappa_s.
    \label{eq:omega_definition}
\end{equation}
Thus,
\begin{equation}
    B_{\mathrm{t}}(x)
    =e^{\alpha_{\mathrm{D}}x}
    \left[\omega_w\ell_w(x)+\omega_s\ell_s(x)\right].
    \label{eq:Bt_explicit}
\end{equation}
Differentiating \eqref{eq:optimized_total_power} gives
\begin{align}
    \frac{\mathrm{d}P_{\mathrm{tot}}^{\star}(x;w,s)}{\mathrm{d}x}
    =&\left[
    \frac{G}{\eta_{\mathrm{R}}}
    +\sqrt{\frac{GA_{\mathrm{t}}+\widetilde{\sigma}_{\mathrm{R}}^2}
    {\eta_{\mathrm{B}}\eta_{\mathrm{R}}B_{\mathrm{t}}(x)}}
    \right]B_{\mathrm{t}}'(x).
    \label{eq:total_power_derivative}
\end{align}
The bracketed factor in \eqref{eq:total_power_derivative} is strictly positive. Moreover,
\begin{equation}
    B_{\mathrm{t}}'(x)=e^{\alpha_{\mathrm{D}}x}F_{w,s}(x),
    \label{eq:Bt_derivative}
\end{equation}
where
\begin{align}
    F_{w,s}(x)
    =&\;\omega_w\left[\alpha_{\mathrm{D}}\ell_w(x)+2(x-x_w)\right]
    \nonumber\\
    &+\omega_s\left[\alpha_{\mathrm{D}}\ell_s(x)+2(x-x_s)\right].
    \label{eq:F_definition}
\end{align}
Therefore, every interior stationary point satisfies
\begin{equation}
    F_{w,s}(x)=0.
    \label{eq:stationary_condition}
\end{equation}
The function $F_{w,s}(x)$ is quadratic:
\begin{equation}
    F_{w,s}(x)=a_2x^2+a_1x+a_0,
    \label{eq:F_quadratic}
\end{equation}
with
\begin{align}
    a_2
    &=\alpha_{\mathrm{D}}(\omega_w+\omega_s),
    \label{eq:a2_stationary}
    \\
    a_1
    &=2\omega_w(1-\alpha_{\mathrm{D}}x_w)
    +2\omega_s(1-\alpha_{\mathrm{D}}x_s),
    \label{eq:a1_stationary}
    \\
    a_0
    &=\omega_w\left[\alpha_{\mathrm{D}}(x_w^2+y_w^2+d^2)-2x_w\right]
    \nonumber\\
    &\quad+\omega_s\left[\alpha_{\mathrm{D}}(x_s^2+y_s^2+d^2)-2x_s\right].
    \label{eq:a0_stationary}
\end{align}
For $\alpha_{\mathrm{D}}>0$, the stationary candidates are
\begin{equation}
    x_{\pm}
    =\frac{-a_1\pm\sqrt{a_1^2-4a_2a_0}}{2a_2},
    \label{eq:stationary_roots}
\end{equation}
provided that the discriminant is nonnegative. Only roots lying in the current \rev{fixed-decoding-order interval} are retained. When $\alpha_{\mathrm{D}}=0$, \eqref{eq:stationary_condition} becomes linear and has the unique solution
\begin{equation}
    x_0=\frac{\omega_wx_w+\omega_sx_s}{\omega_w+\omega_s}.
    \label{eq:no_attenuation_stationary_point}
\end{equation}

\begin{theorem}\label{thm:global_position_solution}
Let $\mathcal{C}$ contain: i) the endpoints $0$ and $L$; ii) all valid \rev{decoding-order switching points} in $\mathcal{X}_{\mathrm{o}}$; and iii) all real roots of \eqref{eq:stationary_condition} that lie in their corresponding \rev{fixed-decoding-order intervals}. At a \rev{decoding-order switching point}, both admissible decoding orders are evaluated. Then a globally optimal solution of problem~\eqref{prob:original}, under the adopted maximum-gain ZF beamformers, additive residual-SI model, and ideal SIC, is obtained by
\begin{equation}
    x^{\star}
    \in
    \operatorname*{arg\,min}_{x\in\mathcal{C}}
    P_{\mathrm{tot}}^{\star}(x;w(x),s(x)).
    \label{eq:global_position_solution}
\end{equation}
The remaining variables are recovered from Proposition~\ref{prop:fixed_x_solution}.
\end{theorem}

\begin{IEEEproof}
\rev{The interval boundaries in \eqref{eq:interval_boundaries} partition the compact feasible set $[0,L]$ into finitely many subintervals over which the NOMA decoding order is fixed.} Within each open subinterval, $P_{\mathrm{tot}}^{\star}(x;w,s)$ is continuously differentiable. Hence, a minimum over the closure of each subinterval must occur either at an endpoint or at an interior stationary point. By \eqref{eq:total_power_derivative}--\eqref{eq:stationary_condition}, all interior stationary points are roots of the quadratic $F_{w,s}(x)$. The resulting candidate set, comprising these stationary roots and all interval boundaries, therefore contains at least one global minimizer over $[0,L]$.
%Enumerating these roots and all interval boundaries therefore contains at least one global minimizer over $[0,L]$. 
\end{IEEEproof}

\begin{algorithm}[t]
\caption{Finite-Candidate Global Power Minimization}
\label{alg:global_solution}
\begin{algorithmic}[1]
\STATE Compute the ZF beamformers from \eqref{eq:vr_zf} and \eqref{eq:vd_zf}, and obtain $G$ and $D$.
\STATE Compute $P_{\mathrm{d}}^{\star}$ from \eqref{eq:optimal_direct_power}.
\STATE Solve \eqref{eq:order_switch_quadratic} and retain all real roots in $(0,L)$.
\STATE Use the retained roots and $\{0,L\}$ to construct the \rev{fixed-decoding-order intervals}.
\STATE Initialize $\mathcal{C}$ with all interval boundaries.
\FOR{each \rev{fixed-decoding-order interval}}
    \STATE Determine $(w,s)$ at an interior test point.
    \STATE Form $a_2$, $a_1$, and $a_0$ from \eqref{eq:a2_stationary}--\eqref{eq:a0_stationary}.
    \STATE Compute the real stationary roots from \eqref{eq:stationary_roots}, or \eqref{eq:no_attenuation_stationary_point} when $\alpha_{\mathrm{D}}=0$.
    \STATE Retain only roots lying in the current interval and add them to $\mathcal{C}$.
\ENDFOR
\STATE At every candidate, compute the resource allocation from Proposition~\ref{prop:fixed_x_solution} and evaluate \eqref{eq:optimized_total_power}.
\STATE At each \rev{decoding-order switching point}, evaluate both decoding orders.
\STATE Return the candidate with the minimum consumed power and its associated $P_{\mathrm{r}}^{\star}$, $P_{\mathrm{d}}^{\star}$, $\beta^{\star}$, $a_w^{\star}$, and $a_s^{\star}$.
\end{algorithmic}
\end{algorithm}

The complete finite-candidate global optimization procedure is summarized in Algorithm~\ref{alg:global_solution}.

\begin{remark}
The proposed method does not require a dense position grid or iterative alternating optimization. \rev{The decoding-order switching equation has at most two real roots, producing at most three fixed-decoding-order intervals.} Each interval contributes at most two stationary roots. Thus, the candidate-set cardinality is bounded by a constant independent of the position-search resolution.%Thus, the number of position candidates is bounded by a small constant for the two-user system.
\end{remark}

\begin{remark}
The global-optimality statement applies to problem~\eqref{prob:original} after the maximum-gain ZF beamformers have been adopted. It does not claim global optimality over unrestricted BS beamforming. A joint beamforming design that permits controlled inter-stream interference may provide an additional power reduction, but generally requires iterative optimization.
\end{remark}

\begin{remark}
No heuristic constraint $a_w\geq a_s$ is imposed. The coefficients are optimized directly from the QoS and SIC requirements. If $a_w\geq a_s$ is required, it is equivalent to $q_w\geq q_s$ and introduces an additional active-constraint case in the fixed-position solution.
\end{remark}

\begin{remark}
Practical BS radiated-power, relay output-power, and relay-gain
constraints can be incorporated as
\begin{equation}
P_r+P_d\leq P_{\mathrm B}^{\max},\qquad
P_{\mathrm R}^{\mathrm{out}}\leq P_{\mathrm R}^{\max},\qquad
\beta\leq\beta_{\max}.
\end{equation}
For fixed \((x,w,s)\), using
\(P_r=A_t+B_t(x)/u\) and
\[
P_{\mathrm R}^{\mathrm{out}}
=
GB_t(x)+
\left(GA_t+\widetilde{\sigma}_{\mathrm R}^{2}\right)u,
\]
the feasible amplification interval is
\begin{align}
u_{\min}(x)
&=
\frac{B_t(x)}
{P_{\mathrm B}^{\max}-P_d^\star-A_t},\\
u_{\max}(x)
&=
\min\left\{
\beta_{\max}^2,\,
\frac{P_{\mathrm R}^{\max}-GB_t(x)}
{GA_t+\widetilde{\sigma}_{\mathrm R}^{2}}
\right\}.
\end{align}
\rev{If $P_{\mathrm B}^{\max}>P_d^\star+A_t$, $P_{\mathrm R}^{\max}>GB_t(x)$, and
\(u_{\min}(x)\leq u_{\max}(x)\), the constrained optimum is}
\[
\widehat u^\star(x)
=
\min\left\{
u_{\max}(x),
\max\left\{u^\star(x),u_{\min}(x)\right\}
\right\}.
\]
Since the resulting optimized cost remains increasing in \(B_t(x)\),
the same finite position-candidate structure applies, with an
additional feasibility check and the projected value
\(\widehat u^\star(x)\) used in Proposition~1.
\end{remark}

\section{Benchmark Schemes}
\label{sec:benchmarks}

For a fair comparison, all schemes use the same user positions, target rates, bandwidth, receiver-noise powers, BS antenna configuration, channel realizations, BS amplifier efficiency, and common BS circuit power. Architecture-specific relay and array circuit powers are included only when the corresponding hardware is present.

\subsection{Direct BS ZF-NOMA Transmission}
\label{subsec:benchmark_direct}

In this benchmark, the relay, dielectric waveguide, and PA are removed, and the BS directly serves all three users. Let $\bm{h}_1$, $\bm{h}_2$, and $\bm{h}_{\mathrm{d}}$ denote their direct BS channels. The BS transmits
\begin{equation}
    \bm{x}_{\mathrm{B}}^{\mathrm{DT}}
    =\bm{v}_{\mathrm{g}}\left(\sqrt{q_w}s_w+\sqrt{q_s}s_s\right)
    +\sqrt{P_{\mathrm{d}}}\bm{v}_{\mathrm{d}}^{\mathrm{DT}}s_{\mathrm{d}},
    \label{eq:direct_tx_signal}
\end{equation}
where $q_w=P_{\mathrm{g}}a_w$ and $q_s=P_{\mathrm{g}}a_s$.

\rev{For this low-complexity direct-transmission benchmark, the NOMA-group beam is nulled at $U_{\mathrm{d}}$.} Define
\begin{equation}
    \bm{R}_{\mathrm{g}}=\bm{h}_1\bm{h}_1^H+\bm{h}_2\bm{h}_2^H.
\end{equation}
The beam $\bm{v}_{\mathrm{g}}$ is the unit-norm principal eigenvector of
\begin{equation}
    \bm{\Pi}_{\mathrm{d}}^{\perp}\bm{R}_{\mathrm{g}}\bm{\Pi}_{\mathrm{d}}^{\perp}.
    \label{eq:direct_group_beam}
\end{equation}
The direct-user beam is nulled at both NOMA users. Let $\bm{H}_{\mathrm{g}}=[\bm{h}_1,\bm{h}_2]$ and
\begin{equation}
    \bm{\Pi}_{\mathrm{g}}^{\perp}
    =\bm{I}-\bm{H}_{\mathrm{g}}
    (\bm{H}_{\mathrm{g}}^H\bm{H}_{\mathrm{g}})^{\dagger}
    \bm{H}_{\mathrm{g}}^H,
\end{equation}
where $(\cdot)^{\dagger}$ denotes the Moore–Penrose pseudo-inverse. Then
\begin{equation}
    \bm{v}_{\mathrm{d}}^{\mathrm{DT}}
    =\frac{\bm{\Pi}_{\mathrm{g}}^{\perp}\bm{h}_{\mathrm{d}}}
    {\|\bm{\Pi}_{\mathrm{g}}^{\perp}\bm{h}_{\mathrm{d}}\|}.
    \label{eq:direct_user_beam}
\end{equation}
This construction requires a nonzero projected direction, typically $N_{\mathrm{t}}\geq 3$ for independent channels.

Define
\begin{equation}
    g_k=|\bm{h}_k^H\bm{v}_{\mathrm{g}}|^2,
    \qquad
    g_{\mathrm{d}}=|\bm{h}_{\mathrm{d}}^H\bm{v}_{\mathrm{d}}^{\mathrm{DT}}|^2.
\end{equation}
The strong user is selected according to $g_s/\sigma_s^2\geq g_w/\sigma_w^2$. Owing to the ZF construction, the SINRs are
\begin{align}
    \gamma_w^{\mathrm{DT}}
    &=\frac{q_wg_w}{q_sg_w+\sigma_w^2},
    \\
    \gamma_{s\rightarrow w}^{\mathrm{DT}}
    &=\frac{q_wg_s}{q_sg_s+\sigma_s^2},
    \\
    \gamma_s^{\mathrm{DT}}
    &=\frac{q_sg_s}{\sigma_s^2},
    \\
    \gamma_{\mathrm{d}}^{\mathrm{DT}}
    &=\frac{P_{\mathrm{d}}g_{\mathrm{d}}}{\sigma_{\mathrm{d}}^2}.
\end{align}
Consequently,
\begin{align}
    q_s^{\star}
    &=\Gamma_s\frac{\sigma_s^2}{g_s},
    \\
    q_w^{\star}
    &=\Gamma_w\left(q_s^{\star}+\frac{\sigma_w^2}{g_w}\right),
    \\
    P_{\mathrm{d}}^{\star,\mathrm{DT}}
    &=\Gamma_{\mathrm{d}}\frac{\sigma_{\mathrm{d}}^2}{g_{\mathrm{d}}},
\end{align}
and the total consumed power is
\begin{equation}
    P_{\mathrm{tot}}^{\mathrm{DT}}
    =\frac{q_w^{\star}+q_s^{\star}+P_{\mathrm{d}}^{\star,\mathrm{DT}}}
    {\eta_{\mathrm{B}}}+P_{\mathrm{c,B}}.
    \label{eq:direct_total_power}
\end{equation}

\subsection{Array-Fed Wi-PASS}
\label{subsec:benchmark_array_pass}

In this benchmark, the horn receiving antenna is replaced by a conventional $N_{\mathrm{R}}$-element receive array. Let
\begin{equation}
    \bm{H}_{\mathrm{BR}}\in\mathbb{C}^{N_{\mathrm{R}}\times N_{\mathrm{t}}}
\end{equation}
denote the BS--relay channel matrix. \rev{The receive combiner and ZF beamformers are constructed as follows.} First, let $\bm{r}_{\max}$ be a dominant right singular vector of $\bm{H}_{\mathrm{BR}}\bm{\Pi}_{\mathrm{d}}^{\perp}$ and define
\begin{equation}
    \bm{v}_{\mathrm{r}}^{\mathrm{arr}}
    =\frac{\bm{\Pi}_{\mathrm{d}}^{\perp}\bm{r}_{\max}}
    {\|\bm{\Pi}_{\mathrm{d}}^{\perp}\bm{r}_{\max}\|}.
    \label{eq:array_relay_beam}
\end{equation}
The maximum-ratio receive combiner is
\begin{equation}
    \bm{q}
    =\frac{\bm{H}_{\mathrm{BR}}\bm{v}_{\mathrm{r}}^{\mathrm{arr}}}
    {\|\bm{H}_{\mathrm{BR}}\bm{v}_{\mathrm{r}}^{\mathrm{arr}}\|},
    \qquad \|\bm{q}\|^2=1.
    \label{eq:array_combiner}
\end{equation}
Define the equivalent first-hop channel $\bm{h}_{\mathrm{eq}}=\bm{H}_{\mathrm{BR}}^H\bm{q}$ and
\begin{equation}
    \bm{v}_{\mathrm{d}}^{\mathrm{arr}}
    =\frac{\bm{\Pi}_{\mathrm{eq}}^{\perp}\bm{h}_{\mathrm{d}}}
    {\|\bm{\Pi}_{\mathrm{eq}}^{\perp}\bm{h}_{\mathrm{d}}\|},
    \qquad
    \bm{\Pi}_{\mathrm{eq}}^{\perp}
    =\bm{I}-\frac{\bm{h}_{\mathrm{eq}}\bm{h}_{\mathrm{eq}}^H}
    {\|\bm{h}_{\mathrm{eq}}\|^2}.
    \label{eq:array_direct_beam}
\end{equation}
The resulting gains are
\begin{equation}
    G_{\mathrm{arr}}
    =|\bm{q}^H\bm{H}_{\mathrm{BR}}\bm{v}_{\mathrm{r}}^{\mathrm{arr}}|^2,
    \qquad
    D_{\mathrm{arr}}
    =|\bm{h}_{\mathrm{d}}^H\bm{v}_{\mathrm{d}}^{\mathrm{arr}}|^2.
    \label{eq:array_gains}
\end{equation}

The beamformers in \eqref{eq:array_relay_beam}-\eqref{eq:array_direct_beam} provide mutual ZF between the direct
user and the relay after receive combining: 
\(\mathbf h_d^H\mathbf v_r^{\rm arr}=0\) and
\(\mathbf q^H\mathbf H_{\rm BR}\mathbf v_d^{\rm arr}=0\).
Moreover, \(\mathbf v_r^{\rm arr}\) maximizes the array first-hop
gain under the former ZF constraint, whereas
\(\mathbf v_d^{\rm arr}\) maximizes the direct-user gain under the
latter constraint.

Let $\widetilde{\sigma}_{\mathrm{R,arr}}^2
=\sigma_{\mathrm{R,arr}}^2+\sigma_{\mathrm{SI,arr}}^2$. Proposition~\ref{prop:fixed_x_solution} and Theorem~\ref{thm:global_position_solution} apply after the substitutions
\begin{equation}
    G\rightarrow G_{\mathrm{arr}},\quad
    D\rightarrow D_{\mathrm{arr}},\quad
    \widetilde{\sigma}_{\mathrm{R}}^2
    \rightarrow\widetilde{\sigma}_{\mathrm{R,arr}}^2.
    \label{eq:array_substitutions}
\end{equation}
For an analog-combining array, its relay circuit power may be modeled as
\begin{equation}
    P_{\mathrm{c,R}}^{\mathrm{arr}}
    =P_{\mathrm{c,R}}
    +N_{\mathrm{R}}P_{\mathrm{LNA}}
    +P_{\mathrm{RF}}
    +N_{\mathrm{R}}P_{\mathrm{PS}}
    +P_{\mathrm{BB}},
    \label{eq:array_circuit_power}
\end{equation}
where $P_{\mathrm{LNA}}$, $P_{\mathrm{RF}}$, $P_{\mathrm{PS}}$, and $P_{\mathrm{BB}}$ denote the low-noise-amplifier, radio-frequency (RF)-chain, phase-shifter, and baseband powers, respectively. This benchmark quantifies the horn antenna's directional-gain, passive-isolation, residual-SI, and hardware-complexity advantages.

\subsection{Horn Relay Without PASS}
\label{subsec:benchmark_horn_nopass}

This benchmark retains the BS and horn-relay architecture but removes the dielectric waveguide and PA. Let the fixed relay transmit location be $\bm{\Phi}_{\mathrm{R}}=(x_{\mathrm{R}},y_{\mathrm{R}},d_{\mathrm{R}})$. The relay--user gain is modeled as
\begin{equation}
    \overline{H}_k
    =G_{\mathrm{R},k}^{\mathrm{t}}G_k^{\mathrm{r}}
    \left(\frac{c}{4\pi f d_{\mathrm{R}k}}\right)^2,
    \qquad k\in\{1,2\},
    \label{eq:nopass_channel}
\end{equation}
where $d_{\mathrm{R}k}=\|\bm{\Phi}_{\mathrm{R}}-\bm{\Phi}_k\|$, and the transmit gain $G_{\mathrm{R},k}^{\mathrm{t}}$ may include the user-dependent horn pattern. \rev{The same channel-gain-to-noise-power ratio $\overline{H}_k/\sigma_k^2$ determines the NOMA decoding order.}

The SINRs follow \eqref{eq:weak_user_sinr}--\eqref{eq:strong_user_sinr} after replacing $H_k(x)$ by $\overline{H}_k$. Since the relay transmit location is fixed, no position optimization is required. Proposition~\ref{prop:fixed_x_solution} applies directly with
\begin{equation}
    B_k=\frac{\sigma_k^2}{G\overline{H}_k}.
    \label{eq:nopass_Bk}
\end{equation}
\rev{In the numerical comparisons, the relay antenna gains, PASS coupling efficiency, and waveguide-injection loss are normalized consistently across the proposed and no-PASS architectures.}

\subsection{Array Relay Without PASS}
\label{subsec:benchmark_array_nopass}

This benchmark combines the array-based wireless first hop in Section~\ref{subsec:benchmark_array_pass} with the fixed conventional relay--user links in Section~\ref{subsec:benchmark_horn_nopass}. Hence, it uses $G_{\mathrm{arr}}$, $D_{\mathrm{arr}}$, $\widetilde{\sigma}_{\mathrm{R,arr}}^2$, and the array circuit power $P_{\mathrm{c,R}}^{\mathrm{arr}}$, while replacing the PASS gains by $\overline H_k$. Proposition~\ref{prop:fixed_x_solution} then applies with
\begin{equation}
    B_k^{\mathrm{arr,noPASS}}
    =\frac{\sigma_k^2}{G_{\mathrm{arr}}\overline H_k}.
\end{equation}
No position optimization is performed. Comparing this scheme with array-fed Wi-PASS isolates the benefit of the waveguide and movable PA under the same array first hop, whereas comparing it with the horn relay without PASS isolates the impact of the relay receiving architecture.

\subsection{Equal-Time OMA-Assisted Wi-PASS}
\label{subsec:oma_benchmark}

In this benchmark, the proposed physical architecture is retained, but users $U_1$ and $U_2$ are served in two equal orthogonal time slots,
\begin{equation}
    \tau_1=\tau_2=\frac{1}{2}.
    \label{eq:oma_equal_time}
\end{equation}
The direct user is served continuously through the ZF direct stream, and one common PA position $x$ is used in both OMA slots. To impose the same average spectral-efficiency requirement $R_k^{\min}=\log_2(1+\Gamma_k)$ as in the NOMA system, the instantaneous SINR target during user $k$'s half-slot is
\begin{equation}
    \Gamma_k^{\mathrm{OMA}}
    =2^{R_k^{\min}/\tau_k}-1
    =2^{2R_k^{\min}}-1
    =(1+\Gamma_k)^2-1.
    \label{eq:oma_target}
\end{equation}

For fixed $x$, define $u_k=\beta_k^2$ and $q_k=P_{\mathrm{r},k}u_k$. The single-user QoS constraint in slot $k$ is
\begin{equation}
    q_k\geq \Gamma_k^{\mathrm{OMA}}
    \left(Au_k+B_k(x)\right),
    \label{eq:oma_constraint}
\end{equation}
which is active at optimum. Following the same argument as Proposition~\ref{prop:fixed_x_solution}, the optimal amplification variable is
\begin{equation}
    u_k^{\star}(x)
    =\sqrt{\frac{\eta_{\mathrm{R}}\Gamma_k^{\mathrm{OMA}}B_k(x)}
    {\eta_{\mathrm{B}}\left[
    G\Gamma_k^{\mathrm{OMA}}A
    +\widetilde{\sigma}_{\mathrm{R}}^2\right]}}.
    \label{eq:oma_optimal_u}
\end{equation}
The minimum variable consumed power during slot $k$ is
\begin{align}
    \psi_k^{\star}(x)
    =&\frac{\Gamma_k^{\mathrm{OMA}}A}{\eta_{\mathrm{B}}}
    +\frac{G\Gamma_k^{\mathrm{OMA}}B_k(x)}{\eta_{\mathrm{R}}}
    \nonumber\\
    &+2\sqrt{\frac{\Gamma_k^{\mathrm{OMA}}B_k(x)
    \left[G\Gamma_k^{\mathrm{OMA}}A
    +\widetilde{\sigma}_{\mathrm{R}}^2\right]}
    {\eta_{\mathrm{B}}\eta_{\mathrm{R}}}}.
    \label{eq:oma_slot_power}
\end{align}
Thus, the equal-time OMA benchmark solves the one-dimensional problem
\begin{equation}
    \min_{0\leq x\leq L}
    \frac{P_{\mathrm{d}}^{\star}}{\eta_{\mathrm{B}}}
    +\frac{1}{2}\sum_{k=1}^{2}\psi_k^{\star}(x)
    +P_{\mathrm{c,B}}+P_{\mathrm{c,R}}.
    \label{eq:oma_problem}
\end{equation}
In the simulations, \eqref{eq:oma_problem} is evaluated over an 81-point uniform position grid. \rev{We intentionally fix $\tau_1=\tau_2=1/2$ rather than optimizing the OMA time fractions. This provides a simple equal-time OMA reference for assessing the effect of simultaneous NOMA transmission; optimizing $\tau_1$ and $\tau_2$ could further reduce the OMA power and is not considered here.}

\subsection{Comparison Objectives and Fairness}

The comparisons quantify: i) the coverage-extension benefit of hybrid Wi-PASS relative to direct BS transmission; ii) the directional, SI-isolation, and hardware advantages of the horn receiver relative to an array-fed relay; iii) the benefit of the dielectric waveguide and movable PA relative to conventional fixed relay--user links; and iv) the multiple-access benefit of NOMA relative to equal-time OMA.

For a fair comparison, the same consumed-power definition is applied to every scheme. The BS radiated power is divided by $\eta_{\mathrm{B}}$ in all cases. Relay output power is divided by $\eta_{\mathrm{R}}$ whenever a relay is present. The same user locations, QoS targets, receiver noise powers, and shadowing realizations are reused for all schemes and all sweep points. Common RF-chain and phase-shifter terms are included consistently, whereas architecture-specific array and relay terms are charged only to the schemes that use them.

\section{Numerical Results}
\label{sec:simulation}

\subsection{Simulation Setup}

We evaluate the proposed scheme and five benchmarks using the same random user and shadowing realizations. The two PASS-side users and the direct user are independently and uniformly distributed over a $30\times 10$~m$^2$ area. The BS is located on the negative $x$-axis at a distance $d_1$ from the waveguide input. Unless otherwise stated, $d_1=50$~m and the common minimum SINR of all three users is $\gamma_0=20$~dB. The principal parameters are listed in Table~\ref{tab:simulation_parameters}.

\begin{table}[t]
\centering
\caption{Simulation Parameters}
\label{tab:simulation_parameters}
\begin{tabular}{ll}
\toprule
Parameter & Value \\
\midrule
Carrier frequency $f$ & $28$ GHz \\
Coverage area & $30\times 10$ m$^2$ \\
Waveguide length/height $(L,d)$ & $(30,3)$ m \\
BS/relay-array elements $(N_{\mathrm t},N_{\mathrm R})$ & $(16,16)$ \\
Bandwidth/noise figure & $400$ MHz/$10$ dB \\
Waveguide power attenuation $\alpha_{\mathrm D}$ & $0.01$ m$^{-1}$ \\
Horn receive gain & $20$ dBi \\
Array element gain & $1$ dBi \\
NLoS path-loss exponent/shadowing std. & $3.5/8$ dB \\
BS/relay power-amplifier efficiencies $(\eta_{\mathrm B},\eta_{\mathrm R})$ & $(0.9,0.9)$ \\
BS/relay fixed circuit powers $(P_{\mathrm{c,B}},P_{\mathrm{c,R}})$ & $(0.1,0.1)$ W \\
Array RF-chain power $P_{\mathrm{RF}}$ & $0.1$ W \\
Array phase-shifter power $P_{\mathrm{PS}}$ & $0.01$ W \\
Horn/array residual SI (baseline) & $-85/-75$ dBm \\
Monte Carlo realizations & $1000$ \\
OMA allocation & $\tau_1=\tau_2=1/2$ \\
OMA position grid & $81$ points \\
\bottomrule
\end{tabular}
\end{table}

The horn first-hop gain follows free-space propagation with the BS array gain and the $20$-dBi horn receive gain. The array-fed benchmark uses a rank-one line-of-sight (LoS) multi-input multi-output (MIMO) first hop with normalized steering vectors and dominant-singular-mode transmission. The direct-transmission links use the NLoS path-loss exponent and log-normal shadowing in Table~\ref{tab:simulation_parameters}. For the residual-SI sweep, the array relay is assigned a 10-dB higher SI power than the horn relay, reflecting the latter's passive directional isolation. All powers reported below are total consumed powers and include the relevant amplifier efficiencies, RF chains, phase shifters, and relay circuitry.

\subsection{Impact of the Minimum SINR}

Fig.~\ref{fig:power_vs_gamma} plots the average consumed power when the common target SINR varies from 10 to 30~dB. \rev{The Monte Carlo averages are reported in Table~\ref{tab:gamma_results}.} Direct transmission is by far the most power demanding because the PASS-side users are served over unfavorable NLoS BS links. Its very large values represent the unconstrained power required to satisfy the prescribed QoS; under practical BS power limits, such operating points would instead be infeasible
Among the relay architectures, adding PASS lowers the second-hop loss, while the horn-fed design further avoids the higher array SI and phase-shifter power. \rev{Consequently, the proposed scheme requires the lowest average consumed power among the considered schemes over the full SINR range.}

\begin{table*}[t]
\centering
\caption{Average Total Consumed Power Versus the Common SINR Target (W)}
\label{tab:gamma_results}
\begin{tabular}{c|rrrrrr}
\toprule
$\gamma_0$ (dB) & Direct & Array Relay without PASS & Horn Relay without PASS & Array-Fed Wi-PASS & Equal-time OMA & Proposed \\
\midrule
10 & 768.52 & 1.7107 & 1.2678 & 1.2010 & 0.85986 & 0.83978 \\
15 & 4488.0 & 9.4042 & 6.8529 & 5.0647 & 3.7014 & 3.2172 \\
20 & 34775 & 81.708 & 59.046 & 40.567 & 30.486 & 24.632 \\
25 & $3.158\times10^5$ & 790.04 & 569.40 & 385.57 & 293.17 & 231.27 \\
30 & $3.057\times10^6$ & 7826.8 & 5636.3 & 3803.8 & 2903.7 & 2273.9 \\
\bottomrule
\end{tabular}
\end{table*}

\rev{The equal-time OMA and NOMA results are close on the logarithmic scale at low SINR, but their linear-power difference becomes substantial as the QoS requirement increases.} Relative to the equal-time OMA benchmark, the proposed NOMA scheme reduces the average consumed power by 2.3\%, 13.1\%, 19.2\%, 21.1\%, and 21.7\% at 10, 15, 20, 25, and 30~dB, respectively. Equivalently, the NOMA--OMA gaps grow from 0.10~dB at 10~dB to 1.06~dB at 30~dB. This trend follows from eq. \eqref{eq:oma_target}: equal-time OMA must support the instantaneous target $(1+\gamma_0)^2-1$ during each half-slot. At $\gamma_0=20$~dB, the proposed scheme also saves 39.3\% relative to array-fed Wi-PASS, 58.3\% relative to the horn relay without PASS, and 69.9\% relative to the array relay without PASS.

\begin{figure}[t]
\centering
\safeincludegraphics[width=0.5\textwidth]{Images/multiuser_power_vs_gamma.eps}
\caption{Average total consumed power versus the common minimum SINR $\gamma_0$.}
\label{fig:power_vs_gamma}
\end{figure}

\subsection{Impact of the BS--Relay Distance}

Fig.~\ref{fig:power_vs_d1} varies $d_1$ from 50 to 200~m at $\gamma_0=20$~dB. The same users and shadowing coefficients are retained at every distance. All wireless-fed schemes require more power as $d_1$ increases because the first-hop gain decays as $d_1^{-2}$. \rev{Nevertheless, the proposed scheme retains the lowest average consumed power throughout the sweep.} From the plotted results, its average power rises by only about 5~dB over the considered range, whereas direct transmission increases by roughly 17~dB because all three direct BS--user distances increase and the two PASS-side users retain the NLoS loss exponent. \rev{Equal-time OMA remains above NOMA, while the two array-based schemes require more power because of their higher circuit consumption and weaker SI suppression. The proposed scheme also requires less power than the Hybrid without PASS scheme over the entire distance range, indicating that the waveguide and movable PA continue to reduce the second-hop power requirement even as the first-hop path loss increases.}

\begin{figure}[t]
\centering
\safeincludegraphics[width=0.5\textwidth]{Images/multiuser_power_vs_d1.eps}
\caption{Average total consumed power versus the BS--relay distance $d_1$ at $\gamma_0=20$~dB.}
\label{fig:power_vs_d1}
\end{figure}

\subsection{Impact of Residual SI}

Fig.~\ref{fig:power_vs_SI} investigates horn-relay SI powers from $-100$ to $-50$~dBm. The array SI is set 10~dB higher at each point. As expected, direct transmission is independent of relay SI and therefore forms a horizontal reference. The relay-assisted curves are nearly insensitive to SI when the impairment is below the thermal-noise-dominated region, but they increase rapidly once SI becomes sufficiently strong. The proposed scheme remains the most power-efficient relay architecture over the complete sweep and is still well below direct transmission at $-50$~dBm. The array-fed schemes deteriorate earlier and more sharply because they combine a 10-dB larger SI power with additional phase-shifter consumption. \rev{Meanwhile, the proposed scheme requires less power than the horn-no-PASS scheme, indicating that the PASS-enabled second hop lowers the consumed power even under strong residual SI. Equal-time OMA follows the same trend but remains consistently higher because of its larger slot-wise SINR targets.}

\begin{figure}[t]
\centering
\safeincludegraphics[width=0.5\textwidth]{Images/multiuser_power_vs_SI.eps}
\caption{Average total consumed power versus the residual SI power at the horn relay. The array-relay SI is 10~dB higher.}
\label{fig:power_vs_SI}
\end{figure}

\section{Conclusion}
\label{sec:conclusion}
\rev{We considered total consumed-power minimization for a hybrid Wi-PASS that serves one direct user and two NOMA users through a full-duplex relay and a movable PA. After adopting maximum-gain ZF beamforming, the resource allocation for any fixed PA position and decoding order reduces to a strictly convex scalar problem with a closed-form optimum. The remaining PA-position optimization is determined by a finite set of decoding-order boundaries and quadratic stationary points, which yields the global optimum with respect to the remaining decision variables under the adopted ZF structure.}

\rev{Numerical comparisons show that the proposed scheme requires the lowest consumed power among the considered schemes over the SINR, BS--relay-distance, and residual-SI sweeps. Relative to equal-time OMA, its power saving increases from 2.3\% at 10~dB to 21.7\% at 30~dB; at a 20-dB target, it reduces power by 39.3\% relative to array-fed Wi-PASS and by 58.3\% relative to the horn relay without PASS. The distance sweep further shows that PASS assistance continues to reduce consumed power as the wireless first hop weakens, while the SI sweep demonstrates the robustness of the proposed architecture to increasing residual SI and highlights the benefits of the considered horn-assisted, PASS-enabled design. Future work may address imperfect SIC, power-dependent residual SI, multiple PAs and waveguides, discrete PA activation, and joint beamforming beyond the adopted ZF structure.}

% Reviewer note: in biblio.bib, capitalize ``nextG'' in the title of the Wi-PASS reference (key: wijewardhana2025wipass), e.g., protect it as {nextG}.
\bibliographystyle{IEEEtran}
\bibliography{biblio}

\balance
\end{document}